\documentclass[submission,copyright,creativecommons]{eptcs}
\providecommand{\event}{FROM 2019} 
\usepackage{breakurl}             
\usepackage{underscore}           

\usepackage{amssymb,amsmath}
\usepackage{mathrsfs,color}
\usepackage{amsthm}
\usepackage{mathtools}
\usepackage{cancel}
\usepackage{xspace}
\DeclareRobustCommand\mos[1]{\mathrel{|}\joinrel
\stackrel{#1}{\mathrel{=}}}

\DeclareRobustCommand\vds[1]{\,\mathrel{|}\joinrel\joinrel\joinrel
\frac{#1}{ \ \ \ }}

\DeclarePairedDelimiter{\ceil}{\lceil}{\rceil}
\DeclarePairedDelimiter{\floor}{\lfloor}{\rfloor}

\newcommand{\dia}[1]{\Diamond_{#1}}
\makeatletter

\newcommand{\tr}{\sigma}
\newcommand{\logl}{\mathbf{K}\Lambda}

\newcommand{\ml}{ML\xspace}
\newcommand{\pl}{HModL\xspace}
\newcommand{\f}{\mathcal{F}}

\DeclareRobustCommand\mos[1]{\mathrel{|}\joinrel
\stackrel{#1}{\mathrel{=}}}
\newcommand{\mosn}[1]{\not \mos{#1}}

\newcommand{\mosm}[1]{\mos{#1}_{{ }_{\ml}}}

\newcommand{\mosp}[1]{\mos{#1}_{{ }_{\pl}}}

\DeclareRobustCommand\vds[1]{\,\mathrel{|}\joinrel\joinrel\joinrel
\frac{#1}{ \ \ \ }}

\newcommand{\mb}{\scriptscriptstyle{\Box}}

\newcommand{\vp}{\varphi}
\newcommand{\srb}{\sigma}
\newcommand{\ri}{\rightarrow}
\newcommand{\spt}{\sigma^{\mb}}

\usepackage{bm,listings,todonotes}

\newtheorem{theorem}{Theorem}
\newtheorem{lemma}[theorem]{Lemma}

\newtheorem{remark}[theorem]{Remark}

\newtheorem{proposition}[theorem]{Proposition}
\newtheorem{definition}[theorem]{Definition}

\providecommand{\urlalt}[2]{\href{#1}{#2}}
\providecommand{\doi}[1]{doi:\urlalt{http://dx.doi.org/#1}{#1}}

\title{From Hybrid Modal Logic to Matching Logic and Back \thanks{All authors contributed equally to this work.}}
\author{Ioana Leu\c stean \and
Natalia Moang\u a \and
Traian Florin \c Serb\u anu\c t\u a 
\institute{
Faculty of Mathematics and Computer Science,\\ University of Bucharest, Str. Academiei 14, 010014 Bucharest, Romania
  \email{ioana@fmi.unibuc.ro, natalia.moanga@drd.unibuc.ro, traian.serbanuta@fmi.unibuc.ro }
}}

\def\titlerunning{From Hybrid Modal Logic to Matching Logic and Back}
\def\authorrunning{I. Leu\c stean, N. Moang\u a \&
T.F. \c Serb\u anu\c t\u a }	
\begin{document}
\maketitle

\begin{abstract}
Building on our previous work on hybrid polyadic modal logic we identify modal logic equivalents for Matching Logic, a logic for program specification and verification. This provides a rigorous way to transfer results between the two approaches, which should benefit both systems. 
\end{abstract}

\section{Introduction}

In this paper, we continue our work from~\cite{noi,noi2}, where we defined a (hybrid) many-sorted polyadic modal logic, for which we proved soundness and completeness, generalizing well-known results from the mono-sorted setting \cite{mod}.

Our research was inspired by Matching logic~\cite{rosu} which made some connections with modal logic~\footnote{Note that Matching logic was further developed in \cite{rosulics}, where techniques from modal logic are employed for the theoretical development.}. Nevertheless, while the system we proposed in \cite{noi2} was strong enough for performing specification and formal verification, its connection with Matching logic, its original motivation, was still to be established.

The purpose of this paper is that of stating the relation between our modal-logic-based systems and Matching logic. In this way we provide a rigorous way to transfer the results between the two approaches, hopefully in the benefit of both systems.  To this aim, we make the following contributions:

(1) We isolate  ${\mathcal H}_{ \Sigma}(\forall)$, a fragment of the system presented in \cite{noi2}, and we show that, when restricted to global deduction, it is equivalent with Matching Logic without definedness.

(2) We introduce ${\mathcal H}_{ \Sigma}(@_z,\forall)$, a strengthening of the system from \cite{noi2} which allows the satisfaction operators $@^s_z$ to also range over state variables, and we show that, when restricted to global deduction, it is equivalent to Matching Logic with the definedness operator.

\paragraph{Background.} For a general background on modal logic we refer to \cite{mod}. We recall that hybrid logics are modal logics that have special symbols (called ``nominals") that name the particular states of a model.  Recall that the satisfaction in modal logic is local, i.e. one analyzes what happens in a given point of the model. With respect to this, nominals can be seen as local constants and, given a model (a frame and an evaluation), the value of a nominal is a fixed singleton set. State variables are variables that range over the individual points of a model, while the usual (propositional) variables range over arbitrary sets of points. All these notions will be detailed in our many-sorted context, but we refer to \cite{hand} for a basic introduction in hybrid modal logics.

For $(S,\Sigma)$ a many-sorted signature, the many-sorted polyadic modal logic ${\mathcal H}_{ \Sigma}$ defined in \cite{noi}  is recalled in Figure~\ref{fig:k}. The system  ${\mathcal H}_{ \Sigma}(\forall)$, defined in Section \ref{sec1}, is a fragment of the system introduced by~\cite{noi2} which enriches ${\mathcal H}_{ \Sigma}$ with nominals, state variables and the forall binder.  This system is a many-sorted generalization  of a hybrid modal logic  defined in \cite{hyb}.  The second system,${\mathcal H}_{ \Sigma}(@_z,\forall)$ is an enrichment of the first one, through the incorporation of the modal satisfaction operators $@^s_x$ for $s\in S$ and $x$ a state variable or a nominal. Intuitively, the operator $@^s_z$ allows us to ``jump" at the element(world, state) denoted by $z$ and the truth value we infer at this point is visible on all sorts. One can see \cite{hand} for a discussion on the expressivity of satisfaction operators in hybrid modal logic. 

The systems $\mathcal{H}_{\Sigma} (\forall)$ and $\mathcal{H}_{\Sigma} (@_z,\forall)$ are presented in Section~\ref{sec1} and Section~\ref{sec2}, along with their completeness results, while the connection with Matching logic is clarified in Section~\ref{sec3}.

\section{The many-sorted hybrid  modal logic ${\mathcal H}_{ \Sigma}(\forall)$}\label{sec1}

Let $(S,\Sigma)$ be a many-sorted signature. In this section we perform hybridization on top of ${\mathcal H}_{ \Sigma}$, the many-sorted polyadic modal logic  defined in \cite{noi}. 

We recall that our language is determined by an $S$-sorted set of  propositional variables ${\rm PROP}=\{{\rm PROP}_s\}_{s\in S}$ such that  ${\rm PROP}_s\neq \emptyset$ for any $s\in S$ and  ${\rm PROP}_{s_1} \cap {\rm PROP}_{s_2} = \emptyset $ for any  $s_1 \neq s_2$ in $S$.
For any $n\in {\mathbb N}$ and $s,s_1,\ldots, s_n\in S$, we denote $\Sigma_{s_1\ldots s_n,s}=\{\sigma\in \Sigma\mid \sigma:s_1 \cdots s_n\to s\}$. The formulas of ${\mathcal H}_{ \Sigma}$ are an $S$-sorted set defined by: 
\vspace*{-0.2cm}
\begin{center}
$\phi_s :=  p\mid j\mid y_s\mid \neg \phi_s \mid\phi_s \vee \phi_s \mid \sigma(\phi_{s_1}, \ldots, \phi_{s_n})_s.$
\end{center}
\vspace*{-0.2cm}
\noindent For any $\sigma\in \Sigma_{s_1 \ldots s_n,s}$ the {\em dual operation} is
$\sigma^{\mb} (\phi_1, \ldots , \phi_n ) := \neg\sigma (\neg\phi_1, \ldots , \neg\phi_n ).$

In general, the sort of a formula will be determined by its context. When necessary, we'll denote by $\vp_s$ the fact that the formula $\vp$ has the sort $s$. 

 In order to define the semantics we introduce $(S,\Sigma)$\textit{-frames} and $(S,\Sigma)$\textit{-models}. An $(S,\Sigma)${\em -frame} is a tuple 
 $\mathcal{F} =({W},(R_\sigma)_{\sigma\in\Sigma})$
where  ${W} =\{ W_s \}_{s\in S}$ is an  $S$-sorted set (whose elements are referred as points, worlds, states, etc.) such that  $W_s\neq \emptyset$ for any $s\in S$ ,
and  ${R}_{\sigma} \subseteq  W_s \times W_{s_1} \times \ldots \times W_{s_n}$  for any $\sigma \in \Sigma_{s_1 \cdots s_n,s}$.
An $(S,\Sigma)$-{\em model based on} $\mathcal{F}$ is a pair  ${\mathcal M}= ({\mathcal F},V)$ where $V =\{V_s\}_{s\in S}$ such that $V_s : \rm{PROP}_s \to \mathcal{P}(W_s)$ for any $s\in S$. The model $\mathcal{M}= (\mathcal{F},V)$ will be simply denoted as  $\mathcal{M}= ({W}, (R_\sigma)_{\sigma\in\Sigma},V)$. For $s\in S$, $w\in W_s$ and  $\phi$ a formula of sort $s$, the  many-sorted  \textit{satisfaction relation} $\mathcal{M},w\mos{s} \phi$
is  defined by structural induction of formulas (see \cite{noi} for details). Moreover, let $s\in S$ and assume $\phi$ is a formula of sort $s$. Then $\phi$ is {\em satisfiable} if ${\mathcal M},w\mos{s}\phi$ for some model $\mathcal M$ and some $w\in W_s$.  The formula $\phi$ is {\em valid} in a model $\mathcal M$ if ${\mathcal M},w\mos{s}\phi$ for any $w\in W_s$; in this case we write ${\mathcal M}\mos{s}\phi$. 

The deductive system of ${\mathcal H}_{ \Sigma}$ is recalled in \ref{fig:k} and the completeness theorem is proved in \cite{noi}.

The hybridization of our many-sorted modal logic is developed  using a combination of ideas and techniques from \cite{hand,pureax,hyb,mod,goranko,goranko2}, but for this section we drew our inspiration mainly form \cite{hyb}. We refer to \cite{noi2} for some similar proofs of the results presented in this section.

Hybrid logic is defined on top of modal logic by  adding {\em nominals}, {\em states variables} and specific operators and binders.  Nominals  allow us to directly refer the worlds (states) of a model, since they are evaluated in singletons in any model. However, a nominal may refer different worlds in different models.
The sorts will be denoted by $s$, $t$, $\ldots$ and by ${\rm PROP}=\{{\rm PROP}_s\}_{s\in S}$, ${\rm NOM}=\{{\rm NOM}_s\}_{s\in S}$ and ${\rm SVAR}=\{{\rm SVAR}_s\}_{s\in S}$ we will denote some countable $S$-sorted sets. The elements of ${\rm PROP}$ are ordinary propositional variables and they will be denoted $p$, $q$,$\ldots$; the elements of ${\rm NOM}$ are called {\em nominals} and they will be denoted by $j$, $k$, $\ldots$; the elements of ${\rm SVAR}$ are called {\em state variables} and they are denoted $x$, $y$, $\ldots$.  We shall assume  that for any distinct sorts $s\neq t\in S$, the corresponding sets of propositional variables, nominals and state variables are distinct. A {\em state symbol} is a nominal or a state variable.

\begin{definition}[${\mathcal H}_{ \Sigma}(\forall)$ formulas] For any $s\in S$ we define the formulas of sort $s$:
 \vspace*{-0.2cm}
\begin{center}
$\phi_s :=  p\mid j\mid y_s\mid \neg \phi_s \mid\phi_s \vee \phi_s \mid \sigma(\phi_{s_1}, \ldots, \phi_{s_n})_s
  \mid  \forall x_t\, \phi_s$
\end{center}
\noindent Here, $p\in {\rm PROP}_s$, $j\in {\rm NOM}_s$, $t\in S$, $x\in {\rm SVAR}_t$, $y\in {\rm SVAR}_s$ and  $\sigma\in \Sigma_{s_1\cdots s_n,s}$. We also define the {dual binder} $\exists$: for any $s,t\in S$, if $\phi$ is a formula of sort  $s$ and $x$ is a state variable of sort $t$, then  $\exists x\, \phi := \neg\forall x\, \neg\phi$ is a formula of sort $s$.  The notions of {\sf free state variables} and {\sf bound state variables} are defined as usual.  
\end{definition}

In order to define the semantics for ${\mathcal H}_{ \Sigma}(@_z,\forall)$ more is needed. Given a model ${\mathcal M}=(W, (R_\sigma)_{\sigma\in\Sigma}, V)$, an {\em assignment} is an $S$-sorted function \mbox{$g : {\rm SVAR} \rightarrow W$.} If $g$ and $g'$ are assignment functions $s\in S$ and \mbox{$x\in \mbox{SVAR}_s$} then we say that $g'$ is an {\em $x$-variant} of $g$ (and we write $g'\stackrel{x}{\sim} g$)  if $g_t=g'_t$ for 
$t\neq s\in S$ and  $g_s(y)=g'_s(y)$ for any \mbox{$y\in \mbox{SVAR}_s$,} $y \neq x$. 

\begin{definition}[The satisfaction relation in  ${\mathcal H}_{ \Sigma}(\forall)$] In the sequel ${\mathcal M}=(W, (R_\sigma)_{\sigma\in\Sigma}, V)$ is a model and $g:{\rm SVAR}\to W$ an $S$-sorted assignment. The satisfaction relation is defined as follows:
\begin{itemize}
\item $\mathcal{M},g,w \mos{s} a$, if and only if $w\in V_s(a)$, where $a\in {\rm PROP_s}\cup {\rm NOM_s}$,
\item $\mathcal{M},g,w \mos{s} x$, if and only if $w=g_s(x)$, where $x\in {\rm SVAR}_s$,
\item $\mathcal{M},g,w \mos{s} \neg \phi$, if and only if $\mathcal{M},g,w \mosn{s}\phi$
\item $\mathcal{M},g,w \mos{s} \phi \vee \psi$, if and only if $\mathcal{M},g,w \mos{s} \phi$ or $\mathcal{M},g,w \mos{s} \psi$ 
\item if $\sigma\in\Sigma_{s_1\ldots s_n,S}$ then $\mathcal{M},g,w \mos{s} \sigma(\phi_1, \ldots , \phi_n )$, if and only if there is \\$(w_1,\ldots,w_n) \in W_{s_1}\times\cdots\times W_{s_n}$ such that  $R_{\sigma} ww_1\ldots w_n$ and $\mathcal{M},g,w_i  \mos{s_i} \phi_i$ for any $i \in [n]$,
 \item $\mathcal{M},g,w \mos{s} \forall x\,\phi$, if and only if  $\mathcal{M},g',w \mos{s} \phi$ for all $g'\stackrel{x}{\sim} g$.

Consequently, 
\item $\mathcal{M},g,w \mos{s} \exists x\, \phi$, if and only if $\exists g'( g' \stackrel{x}{\sim} g \ and \  \mathcal{M},g',w \mos{s} \phi)$.
\end{itemize}
 \end{definition}
 
 In order to define the axioms of our system, one more definition is needed.

We assume  $\#_s$ be a new propositional variable of sort $s$ and  
we inductively define $NC=\{NC_s\}_s$ by
\begin{itemize}
\item $\#_s,\top_s\in NC_s$ for any $s\in S$
\item if $\sigma\in \Sigma_{s_1\cdots s_n,s}$  and $\eta_i\in NC_{s_i}$ for any $i\in [n]$ then $\sigma(\eta_1,\ldots, \eta_n)\in NC_s$.
\end{itemize}

\noindent We further define $NomC=\{NomC_s\}_{s\in S}$ such that $\eta\in NomC_s$ iff $\eta\in NC_s$ and $|\{\#_s\mid s\in S, \#_s\,\in\,\, \eta\}|=1$. If  $\eta\in NomC_{s}$ then $\eta^{\mb}$ is its dual and  $\eta(\vp) \,:=\, \eta [\vp /\#_{s'}]$.

\begin{remark} If $\eta\in NomC_{s}$ and $\vp\in Form_{s'}$ then 
$\mathcal{M},g,w \mos{s}\eta(\vp)  \ \mbox{iff} \ \mathcal{M},h,w' \mos{s'} \vp$ for some $w'$ in the submodel generated by $\mathcal X$ where ${\mathcal X}_s=\{w\}$ and ${\mathcal X}_{t}=\emptyset$ for $t\neq s$. Dually, $\mathcal{M},g,w \mos{s}\eta^{\mb}(\vp)  \ \mbox{iff} \ \mathcal{M},h,w' \mos{s'} \vp$ for any $w'$ in the submodel generated by $\mathcal X$. 
\end{remark}

The deductive system is presented in Figure \ref{fig:k}.

\begin{figure}[h]
\centering
{\small
{\bf The system   ${\mathcal H}_{ \Sigma}$}
\begin{itemize}
\item For any $s\in S$, if $\phi$ is a formula of sort $s$ which is a theorem in propositional logic, then $\phi$ is an axiom. 
\item Axiom schemes: for any $\sigma\in \Sigma_{s_1\cdots s_n,s}$ and for any formulas $\phi_1,\ldots, \phi_n,\phi,\chi$, $\psi$ of appropriate sorts, the following formulas are axioms:

$\begin{array}{rl}
(K_\sigma) & \sigma^{\mb}(\ldots,\phi_{i-1},\phi\rightarrow\chi,\phi_{i+1}, \ldots)\to( \sigma^{\mb}(\ldots ,\phi_{i-1}, \phi, \phi_{i+1},\ldots) \to \sigma^{\mb}(\ldots ,\phi_{i-1}, \chi, \phi_{i+1},\ldots))\\
(Dual_\sigma)& \sigma (\phi_1,\ldots ,\phi_n )\leftrightarrow \neg \sigma^{\mb} (\neg \phi_1,\ldots ,\neg \phi_n )
\end{array}$
\item Deduction rules: \\
\begin{tabular}{rl}
$(MP)$ & if $\vds{s}\phi$ and $\vds{s}\phi\to \psi$ then 
$\vds{s}\psi$\\
$(UG)$ &  if $\vds{s_i}{\phi}$ then $\vds{s}\sigma^{\mb} (\phi_1, .. ,\phi, ..\phi_n)$
\end{tabular}
\end{itemize}

{\bf The system   ${\mathcal H}_{ \Sigma}(\forall)$}

\begin{itemize}
\item The axioms and the deduction rules of ${\mathcal K}_{\Sigma}$
\item Axiom schemes: for any $\sigma\in \Sigma_{s_1\cdots s_n,s}$ and for any formulas $\phi_1,\ldots, \phi_n,\phi,\psi$ of appropriate sorts, the following formulas are axioms:

$\begin{array}{rl}
(Q1) & \forall x\,(\phi \to \psi) \to (\phi \to \forall x\, \psi)
  \mbox{ where $\phi$ contains no free occurrences of x}\\
(Q2) & \forall x\,\phi \to \phi[y\slash x] 
  \mbox{ where $y$ is substitutable for $x$ in $\phi$}\\
(Name) & \exists x \, x \\ 
(Barcan) & \forall x \,\sigma^{\mb}(\phi_1, \ldots, \phi_n) \to \sigma^{\mb}(\phi_1, \ldots,\forall x\phi_{i},\ldots, \phi_n)\\
(Nom) & \forall x\, [\eta(x\wedge\phi)\to \theta^{\mb}(x\to\phi)] $, for any $s\in S$, $\eta$ and $\theta\in {NomC}_s$, $x\in {\rm SVAR}_{s'} 
\end{array}$
\item Deduction rules: \\
\begin{tabular}{rl}
$(Gen)$ & if $\vds{s}\phi$ then $\vds{s}\forall x\phi$, where $\phi\in Form_s$ and $x\in SVAR_t$ for some $t\in S$.
\end{tabular}
\end{itemize} }
\caption{$(S,\Sigma)$ hybrid logic}\label{fig:k}
\end{figure}
\medskip
\textit{Note}: The proofs for the following lemmas: Agreement Lemma, Substitution Lemma, Generalization on nominals are similar to the ones in \cite{noi2}.

\begin{lemma}[Agreement Lemma]
\label{lem:agree} 
Let $\mathcal{M}$ be a standard model. For all standard
$\mathcal{M}$-assignments $g$ and $h$, all states $w$ in $\mathcal{M}$ and all formulas $\phi$ of sort $s \in S$, if $g$ and $h$ agree on all state variables occurring freely in $\phi$, then:
$\mathcal{M},g,w \mos{s} \phi \ \mbox{iff} \ \mathcal{M},h,w \mos{s} \phi$
\end{lemma}

\begin{lemma}[Substitution Lemma]\label{lem:subst}
Let $\mathcal{M}$ be a standard model. For all standard
$\mathcal{M}$-assignments $g$, all states $w$ in $\mathcal{M}$ and all formulas $\phi$, if $y$ is a state variable that is substitutable for $x$ in $\phi$ and $j$ is a nominal then:
\begin{itemize}
\item \label{hunu}$\mathcal{M},g,w \mos{s} \phi[y/x]$ iff $\mathcal{M},g',w \mos{s} \phi$ where $g' \stackrel{x}{\sim} g$ and $ g'_s(x) =g_s(y)$
\item \label{hdoi} $\mathcal{M},g,w \mos{s} \phi[j/x]$ iff $\mathcal{M},g',w \mos{s} \phi$ where $g' \stackrel{x}{\sim} g$ and $ g'_s(x) =V_s(j)$
\end{itemize}
\end{lemma}

%
%
%

\begin{lemma}[Generalization on nominals]\label{lem:gennom}
Assume \mbox{$\vds{s}\phi[i/x]$} where $i\in {\rm NOM}_t$ and $x\in {\rm SVAR}_t$ for some $t\in S$. Then there is a state variable $y\in {\rm SVAR}_t$ that does not appear in $\phi$ such that $\vds{s}\forall y\phi[y/x]$
\end{lemma}

Following the construction of the canonical model of {$\mathbf{K}\forall$} we define
$\mathcal{M}^{\mathbf{K}\forall}=({W}^{\mathbf{K}\forall},{R}^{\mathbf{K}\forall}, V^{\mathbf{K}\forall})$  as follows: (1) for any $s\in S$,  ${W}^{\mathbf{K}\forall}_s=\{ \Phi \subseteq Form_s \mid \Phi \mbox{ is maximal ${\mathbf K}\forall$-consistent \ set} \}  $, 
(2) for any  $\srb \in \Sigma_{s_1 \ldots s_n,s}, w\in W^{\mathbf{K}\forall}_s, u_1\in W^{\mathbf{K}\forall}_{s_1},\ldots, u_n\in W^{\mathbf{K}\forall}_{s_n}$ we define ${R}^{\mathbf{K}\forall}_{\srb} wu_1\ldots u_n$ iff $\srb(\psi_1,\ldots , \psi_n)\in w$ implies $\psi_1 \in u_1, \ldots, \psi_n\in u_n$, (3) for every propositional symbol or nominal $a$, $V^{\mathbf{K}\forall}= \lbrace V^{\mathbf{K}\forall}_s\rbrace_{s\in S}$ is the valuation defined by 
$V^{\mathbf{K}\forall}_s(a) = \lbrace w \in W^{\mathbf{K}\forall}_s |\  a\in w \rbrace$ for any $s\in S$. Note that  $V^{\mathbf{K}\forall}_s(a)$ might be empty or might contain more that one element. We address these issues in the rest of this section.

\begin{definition}[Witnessed Sets]
 Let $s\in S$ and $\Gamma_s$ a maximal ${\mathbf K}\forall$-consistent set. $\Gamma_s$ is called witnessed iff for any ${\mathbf K}\forall$-formula of the form $\exists x \phi$ with $\phi\in Form_s$ there is a nominal $j$ having the same sort as $x$ such that $\exists x \phi \ri \phi[j/x]\in \Gamma_s$.

\end{definition}
\begin{lemma}[Extended Lindenbaum Lemma]\label{extendLinden}
Let  ${\mathbf K}\forall$ and  ${\mathbf K}\forall^+$ be two countable languages such that  ${\mathbf K}\forall^+$ is  ${\mathbf K}\forall$ extended with a countably infinite set of new nominals. Then every consistent set of  ${\mathbf K}\forall$-formulas, $\Gamma_s$, can be extended to a witnessed maximal  ${\mathbf K}\forall^+$-consistent set, $\Gamma_s^+$. 
\end{lemma}

\begin{proof}
Let $E_n=\{j_1, j_2, j_3 \ldots \}$ be an enumeration of the set of all new nominals that are in  ${\mathbf K}\forall^+$, and let $E_f=\{ \phi_1, \phi_2, \phi_3 \ldots \}$ be an enumeration of all  ${\mathbf K}\forall^+$-formulas. We define inductively the maximal  ${\mathbf K}\forall^+$-consistent set $\Gamma_s^+$ for any $s\in S$.

Let $\Gamma_s^0 =\Gamma_s$. $\Gamma_s^0$ contains no nominals from $E_n$, therefore it is consistent when regarded as a set of  ${\mathbf K}\forall^+$-formulas. To prove this, let us suppose that we can prove $\bot_s $ by making use of nominals from $E_n$, then by replacing all the $E_n$ nominals in such a proof with state variables from  ${\mathbf K}\forall$ , we get a proof of $\bot_s$ in  ${\mathbf K}\forall$ , which is a contradiction. 

We define $\Gamma_s ^n$ as follows. If $\Gamma_s^n\cup \{ \phi_n\}$ is inconsistent, then $\Gamma_s^{n+1}=\Gamma^n$. Otherwise:
\begin{itemize}
\item[1)] $\Gamma_s^{n+1} = \Gamma_s^n \cup \{\phi_n\}$, if $\phi_n$ is not of the form $\exists x \psi$
\item[2)] $\Gamma_s^{n+1} = \Gamma_s^n \cup \{\phi_n\} \cup \{\psi[j/x] \}$, if $\phi_n =\exists x \psi $ and $j$ is the the first nominal in the enumeration $E_n$ which is not used in the definitions of $\Gamma_s^i$ for all $i \leq n$ and also does not appear in $\phi_n$.
\end{itemize}

Let $\Gamma_s^+ =\bigcup_{n\geq 0} \Gamma_s^n$. By construction $\Gamma_s^+$ is maximal and witnessed and we need to prove that it is consistent. Let us suppose that $\Gamma_s^+$ is inconsistent, therefore for some $n \geq 0$, $\Gamma_s^n$ is inconsistent. But we will prove that all $\Gamma_s^n$ are consistent. Hence, we need to prove that expansion using $2)$ preserve consistency. Suppose  $\Gamma_s^{n+1} = \Gamma_s^n \cup \{\phi_n\} \cup \{\psi[j/x] \}$ is inconsistent, where $\phi_n = \exists x \psi$. Then there is a formula $\chi$ which is a conjunction of a finite number of formulas from  $ \Gamma_s^n \cup \{\phi_n\}$, such that \mbox{$\vds{s} \chi \ri \neg \psi[j/x] $.} By Lemma \ref{lem:gennom} we can prove that \mbox{$\vds{s} \forall y( \chi \ri \neg \psi[j/x]) $}, for some state variable $y$ that does not occur in $ \chi \ri \neg \psi[j/x]$. Therefore by $(Q1)$ we get \mbox{$\vds{s}  \chi \ri \forall y\neg \psi[y/x] $.} Hence \mbox{$  \Gamma_s^n \cup \{\phi_n\} \vds{s}   \forall y\neg \psi[y/x] $}, and by Lemma \ref{lem:subst} we obtain \mbox{$  \Gamma_s^n \cup \{\phi_n\} \vds{s}   \forall x\neg \psi $}. But $\phi_n = \exists x \psi$, and this contradicts the consistency of $  \Gamma_s^n \cup \{\phi_n\}$.
\end{proof}

\begin{definition}[Witnessed Models]

Let $\mathcal{M}_{\mathbf{K}\forall}^{wit}$ be the {\em witnessed canonical model} which is defined as the canonical model, but only witnessed maximal consistent sets are considered, i.e. all the relations, as well as the valuation are restricted and co-restricted to witnessed maximal consistent sets. 
\end{definition}

                
\begin{lemma}
Let  $\mathcal{M}^{\mathbf{K}\forall}=({W}^{\mathbf{K}\forall}, {R}^{\mathbf{K}\forall}, V^{\mathbf{K}\forall})$ be a canonical model, $\Upsilon$ be a  witnessed maximal consistent set of sort $s$, where $\Upsilon \in \mathcal{W}^{\mathbf{K}\forall}_s $ and let $\mathcal{M}_{\mathbf{K}\forall}^{wit}= ({W}^{wit},{R}^{wit},{V}^{wit})$ be the witnessed submodel of $\mathcal{M}^{\mathbf{K}\forall}$ generated by $\Upsilon$. For any $t\in S$, any state symbol $x\in {\rm SVAR}_t$ and for all witnessed maximal consistent sets $\Gamma$ and $\Delta$ in $W^{wit}_t$, if $x \in \Gamma \cap \Delta$, the $\Gamma = \Delta$. 
\end{lemma}

\begin{proof}
Suppose that $\Gamma$ and $\Delta$ are different, then there is a formula $\phi$ such that $\phi \in \Gamma$ and $\phi \not \in \Delta$. But $\Delta$ and $\Gamma$ are maximal consistent sets, therefore, we get $\phi \in \Gamma$ and $\neg \phi \in \Delta$. From hypothesis, we have $x \in SVAR_t$, where $x \in \Gamma \cap \Delta$. Thus, $x\wedge \phi\in \Gamma$ and $x \wedge \neg \phi \in \Delta$. Recall that $\Gamma$ and $\Delta$ belong to the generated submodel, therefore, exists $\eta_1, \eta_2 \in NC_s$ such that $\eta_1(x\wedge \phi)\in \Upsilon$ and $\eta_2(x\wedge \neg \phi)\in \Upsilon$. As $\Upsilon$ contains every instance of a $Nom$ schema, for some state variable $y \in SVAR_t$ that does not occur freely in $\phi$, $\forall y (\eta_1(y \wedge \phi) \ri \eta^{\mb}_2(y \ri \phi))\in \Upsilon$. Suppose that $x$ is substitutable for $y$ in $\phi$. By $Q2$, we get $\eta_1(x \wedge \phi) \ri \eta^{\mb}_2(x \ri \phi)\in \Upsilon$. But $ \eta_1(x \wedge \phi)\in \Upsilon$, therefore $\eta^{\mb}_2(x \ri \phi)\in \Upsilon$. So, we have $\neg \eta_2(x \wedge \neg \phi)\in \Upsilon$ and  $ \eta_2(x \wedge \neg \phi)\in \Upsilon$, which contradicts that $\Upsilon$ is a maximal consistent set. We conclude that $\Gamma=\Delta$. 
\end{proof}

Recall that to have a standard model we need a model in which every nominal is true at exactly one state. Until now, from the previous lemma we know that the nominals are contained in at most one maximal consistent set in a witnessed model. Therefore, whenever we have a witnessed model $\mathcal{M}_{\mathbf{K}\forall}^{wit}$ such that some state variable does not occur in any maximal consistent set in  $\mathcal{M}_{\mathbf{K}\forall}^{wit}$, we will complete the model by adding a new dummy state symbol $\star$.

\begin{definition}
Let $\mathcal{M}_{\mathbf{K}\forall}^{wit}=({W}^{wit}, {R}^{wit},{V}^{wit}) $ be a witnessed model generate by the witnessed maximal consistent set $\Upsilon$. For any $t \in S$ and any $x \in SVAR_t$ if there exists a maximal consistent set $\Delta \in W^{wit}_t$ such that $x\in \Delta$, then the completed model  $\mathcal{M}_{\star}^{wit}$ is simply $\mathcal{M}_{\mathbf{K}\forall}^{wit}$. Otherwise, $W^{wit \star}_t = W^{wit}_t \cup \{\star _t\} $ and $R^{wit \star}= R^{wit} \cup \{(\star_t, \Upsilon)\ |\  t\in S, \star_t \in W^{wit \star}_t \}$. For all propositional symbols $p$, $V^{wit \star}_t(p)= V^{wit}_t(p)$ and for all nominals $j$, $V^{wit \star}_t(j)=\{\Gamma_t \in \mathcal{M}_{\mathbf{K}\forall}^{wit}\ |\  j \in \Gamma_t\}$ if this set is not empty, and $V^{wit \star}_t(j)=\{\star \} $ otherwise. For all state variables $x \in SVAR_t$, $g^{wit \star}_t(x)=\{ \Gamma_t \in \mathcal{M}_{\mathbf{K}\forall}^{wit}\ |\  x \in \Gamma_t\}$ if this set is not empty, and $g^{wit \star}_t(x)=\{\star \} $ otherwise.

\end{definition}
%
%
%
%
%
%

\begin{lemma}\label{existsfree}
Let $\phi$ and $\chi$ be formulas and $x$ and $y$ state variables such that $y$ is substitutable for $x$ in $\chi$, and $y$ does not have free occurrences in either $\phi$ or $\chi$. Then for any sort $s \in S$ and any $\vp_i \in Form _{s_i}$, for $i \in [n]$ and $i \not = t$, we have that: \\
$\vds{s} \srb (\vp_1, \ldots,\vp_{{t-1}}, \phi, \vp_{{t+1}}, \ldots, \vp_{n}) \to \exists y \srb (\vp_1, \ldots,\vp_{{t-1}},(\exists x \chi \to \chi[y/x]) \wedge \phi, \vp_{{t+1}}, \ldots, \vp_{n})$.

\end{lemma}

\begin{proof}
The proof is similar to the one in \cite{hyb}.
\end{proof}

\begin{lemma}[Existence Lemma for Witnessed Models]\label{lem:existwit}
Let $w$ be a witnessed maximal consistent set. If $\srb (\phi_1, \ldots, \phi_n) \in w$ then there exists witnessed maximal consistent sets $u_i$ such that ${R}^{\mathbf{K}\forall}_{\srb} wu_1\ldots u_n$ and $\phi_i \in u_i$ for any $i\in [n]$. 
\end{lemma}

\begin{proof}
The proof for unary operators is similar with \cite[Lemma 4.20]{mod} for any sort $s\in S$. We prove this lemma for higher arity and start with $\srb$ a binary operator.

Suppose $\srb(\phi_1, \phi_2) \in w$, where $\phi_1 \in  Form_{s_1}$ and $\phi_2 \in Form_{s_2}$. We define $u_1^{-}:=\{\psi| \spt(\psi,\neg \phi_2)\in w\}$. We prove that $u_1^{-} \cup \{ \phi_1 \}$ is consistent. Let us suppose is not consistent. Then there are formulas of sort $ s_1$, $\psi_1,\ldots,$ $\psi_m \in u_1^{-}$ such that \mbox{$\vds{s_1} \psi_1 \wedge \ldots \wedge \psi_m \to \neg \phi_1$}. Easy modal reasoning yields \mbox{$\vds{s} \spt(\psi_1 \wedge \ldots \wedge \psi_m, \neg \phi_2) \to \spt(\neg \phi_1, \neg \phi_2)$}. But \mbox{$\vds{s} \spt(\psi_1, \neg \phi_2) \wedge \ldots \wedge \spt(\psi_m,$} $ \neg \phi_2) \to \spt(\psi_1 \wedge \ldots \wedge \psi_m, \neg \phi_2)$, so $\vds{s} \spt(\psi_1, \neg \phi_2) \wedge \ldots \wedge \spt(\psi_m, \neg \phi_2) \to \spt(\neg \phi_1, \neg \phi_2)$. We have $\spt(\psi_1, \neg \phi_2) \in w, \ldots,$ $ $ $\spt(\psi_m, \neg \phi_2) \in w$ and $w$ is a witnessed maximal consistent set, thus it follows that $\spt(\neg \phi_1, \neg \phi_2) \in w$. So, we get that $\neg \srb(\phi_1, \phi_2)\in w $, which is a contradiction, since $w$ is consistent. Therefore, $u_1^{-} \cup \{ \phi_1 \}$ is consistent and can be extended by Lindenbaum's Lemma to $u_1$ a maximal consistent set. By construction, $\phi_1 \in u_1$. 
We define $u_2^{-}:=\{\psi_2| \mbox{ exists } \psi_1 \in u_1 \mbox{ such that } \spt(\neg \psi_1, \psi_2)\in w\}$. We prove that $u_2^{-} \cup \{ \phi_2 \}$ is consistent. Let us suppose is not consistent. Then there exists formulas of sort $s_2$, $\psi_2^{1},\ldots,\psi_2^m \in u_2^{-}$ such that $\vds{s_2} \psi_2^{1} \wedge \ldots \wedge \psi_2^{m} \to \neg \phi_2$. Also, because $\psi_2^{1},\ldots,\psi_2^m \in u_2^{-}$,by definition of $u_2^{-}$,  we have that there exists formulas $\psi_1^{1},\ldots,\psi_1^m \in u_1$ such that $\spt(\neg \psi_1^{1},  \psi_2^{1}, \ldots, \spt(\neg \psi_1^{m}, \psi_2^{m} \in w$. Let $\psi := \neg \psi_1^{1}\vee \ldots \vee \neg \psi_1^{m}$. Therefore, we have $\spt(\psi, \psi_2^{1}), \ldots, \spt(\psi,  \psi_2^{m}) \in w$.

Easy modal reasoning applied on $\vds{s_2} \psi_2^{1} \wedge \ldots \wedge \psi_2^{m} \to \neg \phi_2$ yields that $\vds{s} \spt(\psi, \psi_2^{1}\wedge \ldots \wedge \psi_2^m) \to \spt(\psi, \neg \phi_2)$. But $\vds{s} \spt(\psi, \psi_2^{1}) \wedge \ldots \wedge \spt(\psi, \psi_2^m) \to \spt(\psi, \psi_2^1 \wedge \ldots \wedge \psi_2^m)$, therefore $\vds{s} \spt(\psi, \psi_2^{1}) \wedge \ldots \wedge \spt(\psi, \psi_2^m) \to \spt(\psi, \neg \phi_2)$. We have $\spt(\psi, \psi_2^{1}), \ldots, \spt(\psi,  \psi_2^{m}) \in w$ and $w$ is a witnessed maximal consistent set, thus it follows that $\spt(\phi, \neg \phi_2) \in w$. So, by definition of $u_1^{-}$, we get that $\psi \in u_1^{-} \subseteq u_1 $, which is equivalent with $\neg \psi_1^{1}\vee \ldots \vee \neg \psi_1^{m} \in u_1$. Hence, exists $k\in [m]$ such that $\neg \psi_1^{k}\in u_1$. But $\psi_1^{k}\in u_1$ and this contradicts the consistency of $u_1$. Therefore, $u_2^{-} \cup \{ \phi_2 \}$ is consistent and can be extended by Lindenbaum’s Lemma to $u_2$ a maximal consistent set. By construction, $\phi_2 \in u_2$.

Let us verify if $R^{\mathbf{K}\forall}_{\srb}wu_1u_2$. From \cite[Lemma 2.18]{noi} we need to verify that $\spt(\psi_1, \psi_2) \in w$ implies $\psi_1 \in u_1$ or $\psi_2 \in u_2$. Suppose $\spt(\psi_1, \psi_2) \in w$. We have two cases. If $\psi_1 \in u_1$, then we get $R^{\mathbf{K}\forall}_{\srb} wu_1u_2$. If $\psi_1 \not \in u_1$, then $\neg \psi_1 \in u_1$, so $\spt (\neg(\neg \psi_1)), \psi_2) \in w$. By definition of $u_2^{-}$, we can conclude that $\psi_2 \in u_2$.

In the same way we can prove the case for higher arity. Let us suppose than $w$ is a maximal consistent set and $\srb (\phi_1, \ldots, \phi_{n-1}) \in w$ then there exists  maximal consistent sets $u_i$ such that $R^{\mathbf{K}\forall}_{\srb} wu_1\ldots u_{n-1}$ and $\phi_i \in u_i$ for any $i\in [n-1]$ where $u_{n-1}^{-}:=\{\psi_{n-1}| \mbox{ for any } i\in [n-2] \mbox{  there exists } \psi_i \in u_i$ such that $ \spt(\neg \psi_1,\ldots, \neg\psi_{n-2}, \psi_{n-1})\in w \}$.

So, we proved that there exist maximal consistent sets $u_i$. Now we want to prove that we can expand those maximal consistent sets to witnessed maximal consistent sets.

Enumerate all the formulas of form $\exists x \chi$, where $x$ can be any state formula of any sort. For each formula in the enumeration we add a suitable witnessed conditional. In this way we inductively expand each $u_i$ for any $i \in [n]$ to a witnessed maximal consistent set. 
 
Suppose that $\srb:Form_{s_1} \times \cdots \times Form_{s_n} \ri Form_{s}$ and define $\dia{t}(\vp) :=\srb(\vp_1, \ldots, \vp_{t-1}, \vp, \vp_{t+1}, \ldots, \vp_n)$ where $\vp \in Form_{s_t}$. Now we enumerate all the formulas of form $\exists x \chi$ of  sort $s_t$ where $x$ can be any state variable of  any sort. The notation $\omega(\exists x \chi, i)$ stands for the witnessed conditional for $\exists x \chi$ in nominal $i$, in other words the formula $\exists x\chi \ri \chi[i/x]$. Also, we use the notation $u_{t}^0 := u_t$ for the maximal consistent set from which we start to expand it to the needed witnessed maximal consistent set. Suppose that for the firsts $m$ formulas in the enumeration we expanded $u_{t}^0$ to a witnessed maximal consistent set $u_{t}^m$. We shall prove that if $\epsilon_{m+1}$ is the $(m+1)$-formula in the enumeration then it is possible to choose a nominal $j_{m+1}$ such that the set $u_{t}^{m+1}= u_{t}^{m} \cup \{ \omega(\epsilon_{m+1}, j_{m+1}\} $ is consistent. Therefore, we will prove that it is possible to choose $j_{m+1}$ so that $\dia{t}(\vp \wedge\omega(\epsilon_1, j_1)\wedge \ldots \wedge\omega(\epsilon_m, j_m)\wedge \omega(\epsilon_{m+1}, j_{m+1})) \in w$.

As we suppose we have already construct $u_t^m$ a witnessed maximal consistent set which contains the witnessed conditionals $\omega(\epsilon_1, j_1), \ldots ,\omega(\epsilon_m, j_m)$ for the firsts $m$ formulas in the enumeration, such that $\dia{t}(\vp \wedge\omega(\epsilon_1, j_1)\wedge \ldots \wedge\omega(\epsilon_m, j_m)) \in w$.
Let $\phi := \vp \wedge\omega(\epsilon_1, j_1)\wedge \ldots \wedge\omega(\epsilon_m, j_m)$.

Suppose that $\epsilon_{m+1}$ is $\exists x \chi$. By Lemma \ref{existsfree} we have $\vds{s_t} \dia{t} (\phi) \ri \exists y \dia{t}((\exists x \chi \ri \chi[y/x])\wedge \phi)$ where $y$ does not have free occurrences in either $\phi$ or $\chi$. Because $\dia{t}(\phi)\in w$, then so is $ \exists y \dia{t}((\exists x \chi \ri \chi[y/x])\wedge \phi)\in w$. Since $w$ is a witnessed maximal consistent set, then there is a nominal $j_{m+1}$ such that $\dia{t}((\exists x \chi \ri \chi[j_{m+1}/x])\wedge \phi)\in w$.
 Therefore, we chose $\omega(\epsilon_{m+1}, j_{m+1})= \exists x \chi \ri \chi[j_{m+1}/x] $ to be the needed witnessed conditional and we define $u_{t}^{m+1}:=u_{t}^{m} \cup \{ \omega(\epsilon_{m+1}, j_{m+1})\} $.

By construction, we have $\dia{t}(\vp \wedge\omega(\epsilon_1, j_1)\wedge \ldots \wedge\omega(\epsilon_m, j_m)\wedge \omega(\epsilon_{m+1}, j_{m+1})) \in w$. But is $u_{t}^{m+1}$ consistent? Let us suppose that $u_{t}^{m+1}$ is not consistent. Then there is a conjunction $\tau$ in $u_{t}^-$ where  $u_{t}^- =\{\vp \ |\ \mbox{for any } i \in [n], i\neq t \mbox{ there exists } \vp_i \in u_i \mbox{ such that } \spt(\vp_1, \ldots, \vp_{t-1}, \vp, \vp_{t+1}, \ldots, \vp_n) \}$  such that $\vds{s_t} \tau \ri \neg (\vp \wedge\omega(\epsilon_1, j_1)\wedge \ldots \wedge\omega(\epsilon_m, j_m)\wedge \omega(\epsilon_{m+1}, j_{m+1}))$. By modal reasoning, we get $\vds{s} \Box_t(\tau)\ri \Box_t(\neg (\vp \wedge\omega(\epsilon_1, j_1)\wedge \ldots \wedge\omega(\epsilon_m, j_m)$ $\wedge \omega(\epsilon_{m+1}, j_{m+1}))$. From definition of $u_t^-$ we know that $\Box_t(\tau) \in w$, so either $ \Box_t(\neg (\vp \wedge\omega(\epsilon_1, j_1)\wedge \ldots \wedge\omega(\epsilon_m, j_m)\wedge \omega(\epsilon_{m+1}, j_{m+1}))\in w$, equivalent with $ \neg \dia{t}((\vp \wedge\omega(\epsilon_1, j_1)\wedge \ldots \wedge\omega(\epsilon_m, j_m)\wedge \omega(\epsilon_{m+1}, j_{m+1}))\in w$ and this contradicts the consistency of $w$. For any $m\geq 0$, $u_t^m$ is a witnessed consistent set, therefore $\bigcup_{m\geq 0} u_t^m$ is a witnessed consistent set and can be extended by Lindenbaum's Lemma to a maximal consistent set. In this way we get the needed witnessed maximal consistent sets for any sort.\end{proof}

\begin{lemma}[Truth Lemma]\label{truthlemma}
Let $\mathcal{M}$ be a completed model, $g$ a completed $\mathcal{M}$-assignment and $w$ an maximal consistent set. For any sort $s\in S$ and any formula $\phi$ of sort $s$, we have:
\begin{center}
$\phi \in w$ if and only if $\mathcal{M}, g, w \mos{s} \phi$
\end{center}

\end{lemma}

\begin{proof}
We make the proof by structural induction on $\phi$.
\begin{itemize}
\item  $\mathcal{M},g,w \mos{s} a$,where $a\in {\rm PROP}_s\cup {\rm NOM}_s$, iff $w\in V_s(a)$ iff $a\in w$;
\item $\mathcal{M},g,w \mos{s} x$, where $x\in {\rm SVAR}_s$, iff $w=g_s(x)$, iff $ x\in w$;
\item $\mathcal{M},g,w \mos{s} \neg \phi$ iff $\mathcal{M},g,w \not\mos{s}\phi$ iff $\phi \not\in w$ (inductive hypothesis) iff $\neg \phi \in w $ (maximal consistent set);
\item $\mathcal{M},g,w \mos{s} \phi \vee \psi$ iff $\mathcal{M},g,w \mos{s} \phi$ or $\mathcal{M},g,w \mos{s} \psi$  iff $\phi \in w$ or $\psi \in w$ (inductive hypothesis) iff $\phi \vee \psi \in w$;
\item let $\tr \in \Sigma_{s_1 \ldots s_n,s}$ and $\phi=\tr (\phi_1, \ldots , \phi_n )$;
\begin{itemize}
\item[``$\Leftarrow$"] $\mathcal{M},g, w \mos{s} \srb (\phi_1, \ldots , \phi_n )$, if and only if for any $i \in [n]$ there exists $u_i \in W_{s_i}$ such that $\mathcal{M},g,u_i  \mos{s_i} \phi_i$ and  ${R}^{\logl}_{\srb} ww_1\ldots w_n$. Using the induction hypothesis we get $\phi_i\in w_i$ for any  $i \in [n]$. Because no maximal consistent set precedes $\star$, we can conclude that neither $u_i$ is $\star$. Therefore, the successors of $w$ must be themselves maximal consistent sets which satisfy the correspondent $\phi_i$. In the end, by applying the induction hypothesis we get $\phi_i \in u_i$ for any $i \in [n]$. Since ${R}^{\mathbf{K}\forall}_{\srb} wu_1\ldots u_{n}$ by definition we infer that $\phi \in w$.
\item[``$\Rightarrow$"]Suppose $\srb (\phi_1, \ldots , \phi_n ) \in w$. Using Existence Lemma \ref{lem:existwit}, for any $i \in [n]$ there  are $u_i$ witnessed maximal consistent sets such that $\phi_i \in u_i$ and $R wu_1\ldots u_n$.  Using the induction hypothesis we get  $\mathcal{M},g, u_i \mos{s_i} \phi_i$ for any $i \in [n]$, so $\mathcal{M},g, w\mos{s} \phi$.
\end{itemize} 
\item let $\phi=\exists x \psi$
\begin{itemize}
\item[``$\Leftarrow$"] Suppose $\mathcal{M}, g, w \mos{s} \exists x \psi$. Then there exists $s\in \mathcal{M} $ such that $\mathcal{M},g',w \mos{s} \psi)$ where $g'\stackrel{x}{\sim} g$ and $g'(x)=\{ s\}$. Because of the definition of the completed models, we know that either a nominal $j$ or a state variable $y$ is true at a state $s$ with respect to the $\mathcal{M}$-assignment function $g$, even if $s =\star$.
\begin{itemize}
\item[[Case 1]] Suppose $V(i)=\{ s\}$. By Substitution Lemma \ref{lem:subst}, $\mathcal{M}, g, w \mos{s}\psi[j/x]$ and by inductive hypothesis $\psi[j/x] \in w$. By means of contrapositive of axiom $(Q2)$ it follows $\phi \in w$.
\item[[Case 2]] Suppose $g(y)=\{s\}$. Firstly, $y$ may not be substitutable for $x$ in $\psi$, therefore we need to replace all the bounded occurrences of $y$ in $\psi$ by some state variable that does not occur in $\psi$ at all. In this way, we get a new formula which we will name it $\psi'$. By Lemma \ref{lem:subst} it follows that $\psi \leftrightarrow \psi'$ is provable and by soundness we get  that it is valid. Now, we have that $\mathcal{M}, g',w \mos{s} \psi'$ and since $y$ is now substitutable for x in $\psi'$, by clause 1 of Substitution Lemma \ref{hunu} it follows $\mathcal{M}, g',w \mos{s} \psi'[y/x]$. By inductive hypothesis $\psi'[y/x] \in w$ and by applying the contrapositive of the $(Q2)$ axiom, it follows $\exists x \psi' \in w$. But $\exists x \psi \leftrightarrow \exists x \psi'$ is provable, therefore $\exists x \psi \in w$.
\end{itemize}
\item[``$\Rightarrow$"]Suppose $\exists x \psi \in w$. As $w$ is a witnessed maximal consistent sets then there is a nominal $j$ of sort $s$ such that $\psi[j/x] \in w$. By the induction hypothesis $\mathcal{M}, g, w \mos{s} \psi[j/x]$ and by means of contrapositive of axiom $(Q2)$ it follows $\mathcal{M}, g, w \mos{s} \exists x \psi$.\end{itemize}\end{itemize}\end{proof}

\begin{theorem}[Hybrid Completeness]
Every consistent set of formulas is satisfiable.
\end{theorem}

\begin{proof}
The proof is similar to the one in \cite{hyb}.\end{proof}

\section{The many-sorted hybrid  modal logic ${\mathcal H}_{ \Sigma}(@_z,\forall)$}\label{sec2}

        Let $(S,\Sigma)$ be a many-sorted signature. As already announced,  in this section we extend the sistem defined in Section~\ref{sec1} by adding the satisfaction operators $@_z^s$ where $s\in S$ and $z$ is a {\em state symbol}, that is, a nominal or a state variable. The formulas 
of  ${\mathcal H}_{ \Sigma}(@_z,\forall)$ are defined as follows:

\vspace*{-0.2cm}
\begin{center}
$\phi_s :=  p\mid j\mid y_s\mid \neg \phi_s \mid\phi_s \vee \phi_s \mid \sigma(\phi_{s_1}, \ldots, \phi_{s_n})_s
  \mid  \forall x_t\, \phi_s\mid @_z^s\psi_t$
\end{center}
\vspace*{-0.2cm}
Here, $p\in {\rm PROP}_s$, $j\in {\rm NOM}_s$, $t\in S$, $x\in {\rm SVAR}_t$, $y\in {\rm SVAR}_s$, $\sigma\in \Sigma_{s_1\cdots s_n,s}$, $z$ is a state symbol of sort $t$ and 
$\psi$ is a formula of sort $t$.

The satisfaction relation is defined similar with the one in ${\mathcal H}_{ \Sigma}(\forall)$, but we only need to add the definition for $@_z$: $\mathcal{M},g,w \mos{s} @_z^s\phi$ if and only if $\mathcal{M},g,Den_g(z)\mos{t} \phi$ where $z$ is a state symbol of sort $t$ and $\phi$ is a formula of the same sort $t$. Here, $Den_g(z)$ is the denotation of the state symbol $z$ of sort $s$ in a model $\mathcal{M}$ with an assignment function $g$, where $Den_g(z)=V_s(z)$ if $z$ is a nominal, and $Den_g(z)=g_s(z)$ if $z$ is a state variable. 

Let us remark that if $z$ is a nominal, then the satisfaction relation is equivalent with the one in \cite{noi2}: $\mathcal{M},g,w \mos{s} @_z^s\phi$ if and only if $\mathcal{M},g,Den_g(z)\mos{t} \phi$ if and only if $\mathcal{M}, g, v \mos{t} \phi$ where $Den_g(z)=V_t(z)=\{v\}$. 

\begin{figure}[h]
\centering

{\bf The system} ${\mathcal H}_{\Sigma}(@_z,\forall)$
\begin{itemize}
\item The axioms and the deduction rules of ${\mathcal K}_{\Sigma}$

\item Axiom schemes: any formula of the following form is an axiom, where $s,s',t$ are sorts,  $\sigma\in\Sigma_{s_1\cdots s_n,s}$,  $\phi,\psi, \phi_1,\ldots,\phi_n$ are formulas (when necessary, their sort is marked as a subscript), $x$ is state variable and $y$, $z$ are state symbols:

$\begin{array}{rlrl}
(K@) & @_z^{s} (\phi_t \to \psi_t) \to (@_z^s \phi \to @_z^s \psi) &
(Agree) &  @_y^{t}@_z^{t'} \phi_s \leftrightarrow @^t_z \phi_s\\
(SelfDual) & @^s_z \phi_t \leftrightarrow \neg @_z^s \neg \phi_t &
(Intro)  & z \to (\phi_s \leftrightarrow @_z^s \phi_s)\\
(Back) & \sigma(\ldots,\phi_{i-1}, @_z^{s_i} {\psi}_t,\phi_{i+1},\ldots)_s\to @_z^s {\psi}_t & 
(Ref) & @_z^sz_t 
\end{array}$\\
\medskip

$\begin{array}{rl}
(Q1) & \forall x\,(\phi \to \psi) \to (\phi \to \forall x\, \psi)
  \mbox{ where $\phi$ contains no free occurrences of x}\\
(Q2) & \forall x\,\phi \to \phi[y\slash x] 
  \mbox{ where $y$ is substitutable for $x$ in $\phi$}\\
(Name) & \exists x \, x \\ 
(Barcan) & \forall x \,\sigma^{\mb}(\phi_1, \ldots, \phi_n) \to \sigma^{\mb}(\phi_1, \ldots,\forall x\phi_{i},\ldots, \phi_n)\\
(Barcan@) & \forall x \, @_z\phi \to @_z \forall x\,\phi, \mbox{where } x \neq z\\
 (Nom\, x) & @_z x\wedge @_y x \to @_z y 
\end{array}$

\medskip

\item Deduction rules:

\begin{tabular}{rl}
$(BroadcastS)$ & if $\vds{s}@_z^s\phi_t$ then $\vds{s'}@_z^{s'}\phi_t$ \\
 $(Gen@)$& if $\vds{s'} \phi$ then $\vds{s} @_z \phi$, where $z $ and $\phi$ have the same sort $s'$\\
$(Paste0)$ & if $\vds{s} @^s_z (y \wedge \phi) \to \psi$ then $\vds{s} @_z \phi \to \psi$\\ & where $z$ is distinct from $y$ that does not occur in $\phi$ or $\psi$\\
$(Paste1)$ & if $\vds{s} @^s_z \sigma(\ldots, y \wedge \phi,\ldots)  \to \psi$ then $\vds{s} @^s_z \sigma(\ldots, \phi, \ldots) \to \psi$\\ & where $z$ is distinct from $y$ that does not occur in $\phi$ or $\psi$\\
$(Gen)$ & if $\vds{s}\phi$ then $\vds{s}\forall x\phi$, where $\phi\in Form_s$ and $x\in {\rm SVAR}_t$ for some $t\in S$.
\end{tabular}
\end{itemize}
\caption{$(S,\Sigma)$ hybrid logic}\label{fig:unu}
\end{figure}

\textit{Note}: Due to the similarities between ${\mathcal H}_{ \Sigma}(@,\forall)$ and ${\mathcal H}_{ \Sigma}(@_z,\forall)$, the following section will contain only the most distinctive proofs.  

\begin{proposition}[Soundness] The deductive systems for ${\mathcal H}_{ \Sigma}(@_z,\forall)$ from Figure \ref{fig:unu} is sound.
\end{proposition}

\begin{lemma}\label{lem:prop}
Let $\Gamma_s$ be a maximal consistent set that contains a state symbol of sort $s$, and for all state symbols $z$, let $\Delta_z=\{ \phi \mid @_z^s \phi \in \Gamma_s$. Then:
\begin{itemize}
\item[1)] For every state symbol $z$ of sort $s$, $\Delta_z$ is a maximal consistent set that contains $z$.
\item[2)] For all state symbols $z$ and $y$ of same sort, $@^s_z \phi \in \Delta_y$ iff $@^s_z \phi \in \Gamma_s$.
\item[3)] There is a state symbol $z$ such that $\Gamma_s = \Delta_z$.
\item[4)] For all state symbols $z$ and $y$ of same sort, if $z \in \Delta_y$ then $\Delta_z= \Delta_y$.
\end{itemize}
\end{lemma}

\begin{proof}
The proofs are similar to the ones in \cite{hybtemp}. 
\end{proof}

This Lemma gives us the maximal consistent sets needed in the Existence Lemma. We build our models out of named sets, i.e. sets containing nominals, and also these are automatically witnessed, therefore, we don't need to glue a dummy state symbol as we deed in the first system, $\mathcal{H}_{\Sigma}(\forall)$, presented in this paper.
But more is needed in order for our model to support an Existential Lemma. Therefore, we add the $Paste$ rules, as you can see in Figure \ref{fig:unu}. In this setting, the system is still sound as we prove in the following:

Now, let $\mathcal{M}$ be an arbitrary named model.

%
$(Paste0)$ Suppose $\mathcal{M},g,w \mos{s} @_z^s(y \wedge \phi) \to \psi$ iff $\mathcal{M},g,w \mos{s} @_z^s(y \wedge \phi)$ implies $\mathcal{M},g,w \mos{s} \psi$. Hence, ($\mathcal{M},g,v\mos{s'} y \wedge \phi$ where $Den_g(z)=\{v \}$ implies  $\mathcal{M},g,w \mos{s} \psi$) iff ($\mathcal{M},g,v \not\!\!\mos{s'} y$ and $\mathcal{M},g,v\not\!\!\mos{s'} \phi$, where $Den_g(z)=\{v \}$, or  $\mathcal{M},g,w \mos{s} \psi$). It follows that ($\mathcal{M},g,v \not\!\!\mos{s'} y$ or  $\mathcal{M},g,w \mos{s} \psi$) and ($\mathcal{M},g,v\not\!\!\mos{s'} \phi$ or  $\mathcal{M},g,w \mos{s} \psi$), where $Den_g(z)=\{v \}$. Then, ($\mathcal{M},g,v\not\!\!\mos{s'} \phi$ or  $\mathcal{M},g,w \mos{s} \psi$), where $Den_g(z)=\{v \}$. So, $\mathcal{M}, g, w \mos{s} @^s_z \phi \to \psi$.

$(Paste1)$ Suppose $\mathcal{M},g,w \mos{s} @_z^s  \sigma(\psi_1, \ldots, \psi_{i-1}, y \wedge \phi, \psi_{i+1}, \ldots, \psi_n) \to \psi$ iff $\mathcal{M},g,w \mos{s} @_j^s  \sigma(\psi_1, \ldots,$ $\psi_{i-1},  y \wedge \phi, \psi_{i+1}, \ldots, \psi_n)$ implies $\mathcal{M},g,w \mos{s} \psi$. Hence,  $\mathcal{M},g,v\mos{s'}   y \wedge \phi$ where $Den_g(z)=\{v \}$ iff exists $(v_1, \ldots,v_n) \in W_{s_1}\times \ldots\times W_{s_n}$ such that $R_{\sigma}v v_1 \ldots v_i \ldots v_n$ where $Den_g(z)=\{v \}$ and $\mathcal{M},g,v_e \mos{s'}  \psi_e$ for any $e \in [n], e\neq i$ and $  \mathcal{M}, g, v_i \mos{s_i} y \wedge \phi$. Hence,  $  \mathcal{M}, g, v_i \mos{s_i} y $ and  $  \mathcal{M}, g, v_i \mos{s_i}  \phi$, so $Den_g(y)=\{v_i\}$ and  $  \mathcal{M}, g, v_i \mos{s_i} \phi$. Then, if there exists $(v_1, \ldots,v_n) \in W_{s_1}\times \ldots\times W_{s_n}$ such that $R_{\sigma}v v_1 \ldots v_i \ldots v_n$ where $Den_g(z)=\{v \}$ and  $\mathcal{M},g,v_e \mos{s'}  \psi_e$ for any $e \in [n], e\neq i$ and $  \mathcal{M}, g, v_i \mos{s_i} \phi$, these imply $\mathcal{M},g,w \mos{s} \psi$. So, $\mathcal{M},g,v \mos{s'}  \sigma(\psi_1, \ldots, \psi_{i-1}, \phi, \psi_{i+1}, \ldots, \psi_n) $ where $ Den_g(z)=\{v \}$ implies $\mathcal{M},g,w \mos{s} \psi$. In conclusion, $\mathcal{M},g,w \mos{s'}  @_z^s\sigma(\psi_1, \ldots, \psi_{i-1}, \phi, \psi_{i+1}, \ldots, \psi_n) \to \psi$.

\begin{definition}[Named, pasted and $@$-witnessed sets]
Let $s\in S$ and $\Gamma_s$ be a set of formulas of sort $s$ from 
 ${\mathcal H}_{ \Sigma}(@_z,\forall)$. We say that 
 \vspace*{-0.2cm}
 \begin{itemize}
 \item $\Gamma_s$ is {\sf named} if one of its elements is a nominal,
 \item $\Gamma_s$ is {\sf pasted} if it is both {\sf 0-pasted} and {\sf 1-pasted}:
 \begin{itemize}
 \item[{\rm (-)}] $\Gamma_s$ is {\sf 0-pasted} if, for any $t\in S$,  $\sigma\in\Sigma_{s_1\cdots s_n,t}$, $z$ a state symbol of sort $t$, and $\phi$ a formula of sort $s_i$, whenever $@_z^s\phi\in \Gamma_s$ there exists a nominal $j\in {\rm NOM}_{s_i}$ such that $@_z^s\sigma(\ldots, \phi_{i-1},j \wedge \phi,\phi_{i+1},\ldots)\in \Gamma_s$. 
 \item[{\rm (-)}] $\Gamma_s$ is {\sf 1-pasted} if, for any $t\in S$,  $\sigma\in\Sigma_{s_1\cdots s_n,t}$, $z$ a state symbol of sort $t$, and $\phi$ a formula of sort $s_i$, whenever $@_z^s\sigma(\ldots, \phi_{i-1},\phi,\phi_{i+1},\ldots)\in \Gamma_s$ there exists a nominal $j\in {\rm NOM}_{s_i}$ such that $@_z^s\sigma(\ldots, \phi_{i-1},j \wedge \phi,\phi_{i+1},\ldots)\in \Gamma_s$. 
 \end{itemize}
\item $\Gamma_s$ is {\sf $@$-witnessed} if the following two conditions are satisfied:
\begin{itemize}
 \item[{\rm (-)}] for $s',t\in S$ , $x\in {\rm SVAR}_t$, $k\in {\rm NOM}_{s'}$ and any formula $\phi$ of sort $s'$,  whenever $@_k^s\exists x\, \phi\in \Gamma_s$ there exists $j\in {\rm NOM}_t$ such that $@_k^s\phi[j/x]\in\Gamma_s$,
 \item[{\rm (-)}] for any $t\in S$ and $x\in {\rm SVAR}_t$ there is $j_s\in {\rm NOM}_t$ such that $@_{j_x}^s x\in \Gamma_s$. 
 \end{itemize}
 \end{itemize}
 \end{definition}

 \begin{lemma}[Extended Lindenbaum Lemma]\label{lemma:lind}
Let $\Lambda$ be a set of  formulas in the language of ${\mathcal H}_{\Sigma}(@_z,\forall)$  and $s\in S$. Then any consistent set $\Gamma_s$ of formulas of sort $s$ from  ${\mathcal H}_{\Sigma}(@_z,\forall)+\Lambda$  can be extended to a named, pasted and $@$-witnessed maximal consistent set by adding countably many nominals to the language.   
\end{lemma}

\begin{proof}
The proof generalizes to the $S$-sorted setting well-known proofs for the mono-sorted hybrid logic, see \cite[Lemma 7.25]{mod}, \cite[Lemma 3, Lemma 4]{pureax}, \cite[Lemma 3.9]{hyb}. 

For each sort $s\in S$, we add a set of new nominals and enumerate this set. Given a set of formulas $\Gamma_s$, define $\Gamma_s^k$ to be $\Gamma_s \cup \{ k_s\} \cup \{@_{j_x}^s x| \ x \in {\rm SVAR_s} \}$, where $k_s$ is the first new nominal of sort $s$ in our enumeration and $j_x$ are such that if $x$ and $y$ are different state variables of sort $s$ then also $j_x$ and $j_y$ are different nominals of same sort $s$. As showed in \cite{noi2}, $\Gamma_s^k$ is consistent.


Now we enumerate  on each sort $s \in S$  all the formulas of the new language obtained by adding the set of new nominals and define $\Gamma^0 := \Gamma_s^k$. Suppose we have defined $\Gamma^m$, where $m \geq 0$. Let $\phi_{m+1}$ be the $m+1-th$ formula of sort $s$ in the previous enumeration. We define $\Gamma^{m+1}$ as follows. If $\Gamma^{m}\cup \{\phi_{m+1}\}$ is inconsistent, then $\Gamma^{m+1} = \Gamma^{m}$. Otherwise:
\begin{itemize}
\item[(i)] $\Gamma^{m+1} = \Gamma^{m} \cup  \{\phi_{m+1}\} $, if $\phi_{m+1}$ is not of the form $@_z\sigma(\ldots, \varphi, \ldots)$, $@_x x$ or $@_j \exists x\varphi(x)$, where $j$ is any nominal of sort $s''$, $\varphi$ a formula of sort $s''$, $x \in {\rm SVAR_{s''}}$ and $z$ is a state symbol.
\item[(ii)] $\Gamma^{m+1} = \Gamma^{m} \cup  \{\phi_{m+1}\} \cup \{@_x  (k \wedge x ) \} $, if $\phi_{m+1}$ is of the form $@_x x $.
\item[(iii)] $\Gamma^{m+1} = \Gamma^{m} \cup  \{\phi_{m+1}\} \cup \{@_x \sigma(\ldots, k \wedge \phi, \ldots)  \} $, if $\phi_{m+1}$ is of the form $@_x \sigma(\ldots, \varphi, \ldots)$. 
\item[(iv)] $\Gamma^{m+1} = \Gamma^{m} \cup  \{\phi_{m+1}\} \cup \{ @_j \varphi[k/x]\}$, where $\phi_{m+1} $ is of the form $@_j \exists x\varphi(x)$. 
\end{itemize}
In clauses $(ii)$ and $(iii)$, $k$ is the first new nominal in the enumeration that does not occur in $\Gamma^i$ for all $i \leq m$, nor in $@_x \sigma(\ldots, \varphi, \ldots)$.

Let $\Gamma ^+= \bigcup_{n\geq 0} \Gamma^n$. Because $k \in \Gamma^0 \subseteq \Gamma^+$, this set in named, maximal, pasted and $@$-witnessed by construction. We will check if it is consistent for the expansion made in the second, third and fourth items.

Suppose $\Gamma^{m+1} = \Gamma^{m} \cup  \{\phi_{m+1}\} \cup \{@_x(k \wedge x) \} $ is an inconsistent set, where $\phi_{m+1}$ is $@_x x$. Then there is a conjunction of formulas $\chi \in \Gamma^m \cup \{\phi_{m+1}\} $ such that \mbox{$\vds{s} \chi \to \neg @_x(k \wedge x) $} and so \mbox{$\vds{s}@_x(k \wedge x)  \to \neg \chi$.} But $k$ is the first new nominal in the enumeration that does not occur neither in $\Gamma^m$, nor in $@_x x$ and by $Paste0$ rule we get $\vds{s} @_x x \to \neg \chi$. Then $ \vds{s} \chi \to \neg @_x x$, which contradicts the consistency of $\Gamma^m \cup  \{\phi_{m+1}\}$. 

Suppose $\Gamma^{m+1} = \Gamma^{m} \cup  \{\phi_{m+1}\} \cup \{@_x \sigma(\ldots, k \wedge \varphi, \ldots)  \} $ is an inconsistent set, where $\phi_{m+1}$ has the form $@_x \sigma(\ldots, \varphi, \ldots)$. Then there is a conjunction of formulas $\chi \in \Gamma^m \cup \{\phi_{m+1}\} $ such that $\vds{s} \chi \to \neg @_x \sigma(\ldots, k \wedge \varphi, \ldots)$ and so \mbox{$\vds{s} @_x \sigma(\ldots, k \wedge \varphi, \ldots) \to \neg \chi$.} But $k$ is the first new nominal in the enumeration that does not occur neither in $\Gamma^m$, nor in $@_x \sigma(\ldots, \varphi, \ldots)$, therefore, by $Paste1$ rule we get $\vds{s} @_x \sigma(\ldots, \varphi, \ldots) \to \neg \chi$. It follows that $\vds{s} \chi \to \neg @_x \sigma(\ldots, \varphi, \ldots)$, which contradicts the consistency of $\Gamma^m \cup  \{\phi_{m+1}\}$. \end{proof}


\begin{definition}[Named models and natural assignments]\label{def:canonic} For any $s \in S$, let $\Gamma_s$ be a named, pasted and witnessed maximal consistent set and for all state symbols $z$, let $\Delta_z =\{ \varphi \mid @_z^s \varphi \in \Gamma_s \}$. Define $W_s = \{ \Delta_x \mid z$ a state symbol of sort s$\}$. Then, we define $\mathcal{M}=(W, \{ R_{\sigma}\}_{\sigma \in \Sigma})$, the named model generated by the $S$-sorted set $\Gamma =\{ \Gamma_s\}_{s\in S}$, where $R_{\sigma}$ and $V$ are the restriction of the canonical relation and the canonical valuation. We define the natural assignment $g_s:{ \rm SVAR}_s \to W_s$ by $g_s(x) = \{ w \in W_s \mid x \in w\} $.
\end{definition}
\begin{lemma}[Existence Lemma]\label{lem:existwit2}
Let $\mathcal{M}=(W, \{ R_{\sigma}\}_{\sigma \in \Sigma})$ be a named model generated by a named and pasted $S$-sorted set $\Gamma$ and let $w$ be a witnessed maximal consistent set. If $\srb (\phi_1, \ldots, \phi_n) \in w$ then there exists witnessed maximal consistent sets $u_i$ such that ${R}_{\srb} wu_1\ldots u_n$ and $\phi_i \in u_i$ for any $i\in [n]$. 
\end{lemma}

\begin{proof}
Let $\srb (\phi_1, \ldots, \phi_n) \in w$, then $@^s_j\srb (\phi_1, \ldots, \phi_n) \in \Gamma_s$, but $\Gamma_s$ is pasted( then $1-pasted$), so there exists $k_1$ a nominal of sort $s_1$ such that $ @^s_j \srb (\phi_1 \wedge k_1, \ldots, \phi_n) \in \Gamma_s$, so $\srb (\phi_1 \wedge k_1, \ldots, \phi_n) \in \Delta_j=w$.
We want to prove that $\Delta_{k_1}, \ldots, \Delta_{k_n}$ are suitable choices for $u_1, \ldots, u_n$.

Let $\psi_1 \in \Delta_{k_1}$. Then $@_{k_1} \psi_1 \in \Gamma_s$ and by agreement property we get $@_{k_1} \psi_1 \in \Delta_{j}$. But $\vds{s} k_1 \wedge \psi_1 \to @_{k_1} \psi_1$ (instance of Introduction axiom), and by modal reasoning we get $\sigma(@_{k_1} \psi_1, \phi_2, \ldots, \phi_n) \in \Delta_j$. From Back axiom, $@_{k_1} \psi_1 \in \Delta_j$ and by using the agreement property, $@_{k_1} \psi_1 \in \Gamma_s$. Hence, $\psi_1 \in \Delta_{k_1}$. 

Now, $\srb(\psi_1, \phi_2, \ldots, \phi_n) \in \Delta_j$, then $@_j \sigma(\psi_1, \phi_2, \ldots, \phi_n) \in \Gamma_s $, but the set is pasted, then exists $k_2$ a nominal of sort $s_2$ such that $@_j\sigma(\psi_1, k_2 \wedge \phi_2, \phi_3, \ldots, \phi_n) \in \Gamma_s$. Then $\sigma(\psi_1, k_2 \wedge \phi_2, \phi_3, \ldots, \phi_n) \in \Delta_{j}$.

Let $\psi_2 \in \Delta_{k_2}$. Then $@_{k_2} \psi_2 \in \Gamma_s$ and by agreement property we get $@_{k_2} \psi_2 \in \Delta_{j}$. But $\vds{s} k_2 \wedge \psi_2 \to @_{k_2} \psi_2$ (instance of Introduction axiom), and by modal reasoning we get $\sigma(\psi_1, @_{k_2} \psi_2,\phi_3, \ldots, \phi_n) \in \Delta_j$. From Back axiom, $@_{k_2} \psi_2 \in \Delta_j$ and by using the agreement property, $@_{k_2} \psi_2 \in \Gamma_s$. Hence, $\psi_2 \in \Delta_{k_2}$. Therefore, by induction, we get that $\psi_i \in \Delta_{k_i}$ for any $i \in [n]$. Then $@_{k_i} \psi_i \in \Gamma_s$ if and only if, by agreement property, $@_{k_i} \psi_i \in \Delta_j$. But $ \sigma( k_1, \ldots,k_n) \in \Delta_j$ and by using the Bridge axiom, it follows that $\sigma( \psi_1, \ldots, \psi_n) \in \Delta_j$. We proved that for any $i \in [n]$, $\psi_i \in \Delta_{k_i}$ we have $\sigma( \psi_1, \ldots, \psi_n) \in \Delta_j$ and by Definition \ref{def:canonic}, it follows that $R_{\sigma}\Delta_j \Delta_{k_1} \ldots \Delta_{k_n}$.
\end{proof}

\begin{lemma}[Truth Lemma]\label{truthlemma2}
Let $\mathcal{M}$ be a model, $g$ an $\mathcal{M}$-assignment and $w$ a maximal consistent set. For any sort $s\in S$ and any formula $\phi$ of sort $s$, we have
$\phi \in w$ if and only if $\mathcal{M}, g, w \mos{s} \phi$.
\end{lemma}
\begin{proof}
We make the proof by structural induction on $\phi$. All the cases except the one for $@_z$ are similar with the ones of the $\mathcal{H}_{\Sigma}(\forall)$ system. Suppose $\mathcal{M}, g, w \mos{s} @^s _z \phi$ iff $\mathcal{M}, g, \Delta_z \mos{s} @^s _z \phi$ (by Lemma \ref{lem:prop}.(3)) iff $\phi \in \Delta_z$ (inductive hypothesis) iff $@^s_z \phi$ (by $Intro$ axiom together with $z \in \Delta_z$) iff $@^s_z \phi \in w$ (by Lemma \ref{lem:prop}.(2)).
\end{proof}

\begin{theorem}[Hybrid Completeness]
Every consistent set of formulas is satisfied.
\end{theorem}

As in the mono-sorted case,  in ${\mathcal H}_{ \Sigma}(@_z,\forall)$ we  can define the  universal modality: $ A^s\varphi := \forall x @_x^s \varphi$, where $\varphi$ is a formula of sort $t$ and $x \in {\rm SVAR}_t$. Its dual is defined $E^s \varphi = \neg A^s \neg \varphi$.
 
 Note that, in our many-sorted setting, the universal modality has also the role of connecting the sorts (similarly to satisfaction operators).

 \begin{lemma} \label{lem:univ}
Let $\mathcal{M}=(W,\{R_{\sigma}\}_{\sigma \in \Sigma}, V)$ be an $S$-sorted model in ${\mathcal H}_{ \Sigma}(@_z,\forall)$ and $\varphi$ a formula of sort $t$. Then, for any $s \in S$, $ $ $\mathcal{M} \mos{s} A^s \varphi$ iff $\mathcal{M} \mos{t} \varphi$
  \end{lemma}
  
  \begin{proof}
   $\mathcal{M} \mos{s} A^s \varphi$ iff for any $g$, any $w \in W_s$,  $\mathcal{M}, g, w \mos{s} A^s \varphi$
iff for any $g$, any $w \in W_s$,  $\mathcal{M}, g, w \mos{s} \forall x @^s_x \varphi$
iff for any $g$, any $w \in W_s$, any $g'\stackrel{x}{\sim} g$, $\mathcal{M}, g', w \mos{s}  @^s_x \varphi$
iff for any $g$, any $w \in W_s$, any $g'\stackrel{x}{\sim} g$, $\mathcal{M}, g', v \mos{t}  \varphi$ where $g'_t(x)=v$
iff for any $g$, any $w \in W_s$, any $g'\stackrel{x}{\sim} g$, $\mathcal{M}, g', v \mos{t}  \varphi$ for any $v \in W_t$
iff $\mathcal{M} \mos{t} \varphi$.
  \end{proof}

Let $\Gamma=\{\Gamma_s\}_{s\in S}$ an $S$-sorted set of formulas. Then $\mathcal{M} \mos \Gamma$ if and only if $\mathcal{M} \mos{s} \Gamma_s$, for any $s \in S$. We define $\Gamma^A_s = \Gamma_s \cup \{A^s \psi \mid \psi \in \Gamma_t$ for some $t \neq s$\}

\begin{proposition}
Let $\mathcal{M}=(W,\{R_{\sigma}\}_{\sigma \in \Sigma}, V)$ be an $S$-sorted model in ${\mathcal H}_{ \Sigma}(@_z,\forall)$ and $\Gamma=\{\Gamma_s\}_{s\in S}$ an $S$-sorted set of formulas. Let $s \in S$, then $\mathcal{M} \mos{} \Gamma$ if and only if $\mathcal{M} \mos{s}\Gamma^A_s$.
\end{proposition}
\begin{proof}
Suppose $\mathcal{M} \mos{} \Gamma$ if and only if $\mathcal{M} \mos{s} \Gamma_s$ for any $s \in S$. Then $\mathcal{M} \mos{s} \varphi$ for any $\varphi \in \Gamma_s$ for any $s \in S$. Let $s,t \in S$, so for any $\psi \in \Gamma_t$, $\mathcal{M} \mos{t} \psi$ and by Lemma \ref{lem:univ}, we get $\mathcal{M} \mos{s} A^s\psi$, for any $\psi \in \Gamma_t$. It follows that, for any $ \varphi \in \Gamma^A_s$, we have $\mathcal{M} \mos{s} \varphi$ if and only if $\mathcal{M} \mos{s} \Gamma^A_s$. For the right-to-left direction, let $s \in S$ and $\mathcal{M} \mos{s} \Gamma^A_s$. Then, for any $ \varphi \in \Gamma^A_s$, we have $\mathcal{M} \mos{s} \varphi$. If $ \varphi \in \Gamma_s$, then $\mathcal{M} \mos{s} \Gamma_s$. If  $\varphi \in \Gamma^A_s \backslash \Gamma_s$, then $\varphi$ is $A^s \psi$, where $\psi\in \Gamma_t$ for some $t\neq s$ in $S$.  Since $\mathcal{M} \mos{s} A^s \psi$, by  Lemma \ref{lem:univ},  $\mathcal{M} \mos{t} \psi $. Hence  $\mathcal{M} \mos{t} \psi $ for any $\psi \in \Gamma_t$ and  any $t \neq s$ in $S$. It follows that  $\mathcal{M} \mos{t} \Gamma_t$ for any $t\neq s$, so   $\mathcal{M} \mos{} \Gamma$.
\end{proof}

\section{The connection between Matching Logic and Hybrid Modal Logic}\label{sec3}

In this section we analyze the connection between 
Matching logic (\ml) and  the Many-sorted hybrid modal logic (\pl). We denote by \ml a Matching logic system (with and without definedness) and by \pl the corresponding Many-sorted hybrid modal logic system, as follows:
for \ml without definedness, the corresponding system is $\mathcal{H}_{\Sigma}(\forall)$, while
for \ml with definedness, the corresponding system is
$\mathcal{H}_{\Sigma}(@_z, \forall)$.

Recall that a matching logic signature or simply a signature
${\Sigma^{\ml}}= (S,{\rm VAR},\Sigma)$ is a triple with a nonempty set $S$ of sorts, an
$S$-indexed set ${\rm VAR} = \{{\rm VAR}_s\}_{s\in S}$ of countably infinitely many
sorted variables denoted $x:s$; $y:s$, etc., and an $(S^* \times S)$-indexed
countable set $\Sigma = \{\Sigma_{s_1 \ldots s_n,s}\}_{s_1 \ldots s_n,s \in S}$ of many-sorted symbols.

A \textit{matching logic} $ \Sigma^{\ml}$\textit{-model} $M =(\{M_s\}_{s\in S},\ \{\sigma_M\}_{\sigma\in \Sigma})$ consists of a non-empty carrier set $M_s$ for each sort $s\in S$ and a function $\sigma_M : M_{s_1}\times \ldots \times M_{s_n} \rightarrow \mathcal{P}(M_s)$ for each symbol $\sigma \in \Sigma_{s_1 \ldots s_n,s}$ called \textit{the interpretation} of $\sigma$ in $M$. 
 
The pointwise extension, $\sigma_M:\mathcal{P}(M_{s_1})\times \ldots \times \mathcal{P}(M_{s_n}) \rightarrow \mathcal{P}(M_s)$ is defined as: \\
$ {\sigma}_M(A_1, \ldots, A_n) = \bigcup \lbrace   \sigma_m(a_1, \ldots, a_n) | a_i \in A_i \ for \  all \ i \in [n]\rbrace, \mbox{ where } A_i \subseteq M_i  \mbox{ for all }  i \in [n].$
 
Let ${\Sigma^{\ml}}= (S,{\rm VAR},\Sigma)$  and let M be a ${\Sigma^{\ml}}$-model. Given a map $\rho : {\rm VAR} \rightarrow M$, called an $M$\textit{-valuation}, let its extension $\overline{\rho} : PATTERN^{\ml} \rightarrow \mathcal{P}(M)$ be inductively defined as fallows:
\begin{itemize}
\item $\overline{\rho}(x)= \{\rho(x)\}$, for all $x\in {\rm VAR}_s$

\item $\overline{\rho}(\neg \varphi) = M_s \backslash \overline{\rho}(\varphi)$, for all $\varphi \in PATTERN_s$
\item $\overline{\rho}(\varphi_1 \vee \varphi_2) = \overline{\rho}(\varphi_1)\cup \overline{\rho}(\varphi_2)$, for all $\varphi_1, \varphi_2$ patterns of the same sort
\item $\overline{\rho}(\sigma(\varphi_1, \ldots ,\varphi_n)) = \sigma_M ( \overline{\rho}(\varphi_1), \cdots, \overline{\rho}(\varphi_n)),$ for all $\sigma \in \Sigma_{s_1 \ldots s_n,s}$ and appropriate $\varphi_1, \ldots, \varphi_n$

\item $\overline{\rho}(\exists x. \varphi)=\bigcup_{a\in M_{s'}}  \overline{\rho[a/x]}(\varphi) $, for all $x\in {\rm VAR}_s$


\end{itemize}
where ``$\backslash$'' is set difference and $\rho[a/x]$ denotes de $M$-valuation $\rho'$ with $\rho'(x)=a$ and $\rho'(y)=\rho(y)$ for all $y \neq x$.

In Matching logic $M$ \textit{satisfies} $\varphi_s$, written $M \models \varphi_s$, iff $\overline{\rho}(\varphi_s)=M_s$ for all $\rho: {\rm VAR} \rightarrow M$.

For any sorts (not necessarily distinct) $s_1, s_2 \in S$, we consider the unary symbol $\ceil*{ _{-} }^{s_2}_{s_1} \in \Sigma_{s_1, s_2}  $, called the definedness symbol, and the pattern/symbol $\ceil*{ x:s_1 }^{s_2}_{s_1} \in \Sigma_{s_1, s_2} $, called $(Definedness)$.
For all $\rho$, we have $\overline{\rho}(\ceil{\varphi}^{s_2}_{s_1}) = M_{s_2}$ if $\overline{\rho}(\varphi)\neq\emptyset$, and $\overline{\rho}(\ceil{\varphi}^{s_2}_{s_1}) = \emptyset$, otherwise. \textit{Totality}, $\floor{_{-}}^{s_2}_{s_1}$, is defined as a derived construct dual to definedness: $\floor{\varphi}^{s_2}_{s_1} = \neg\floor{\neg \varphi}^{s_2}_{s_1} $.
The following remark clarifies the relation between definedness from Matching Logic and satisfaction operator from our logic.

\begin{remark}
    In a Matching Logic system with a definedness pattern, we can define $@_x^s\phi=\ceil{x\wedge\phi}_{s_x}^s$, while in ${\mathcal H}_{\Sigma}(@_z,\forall)$ we can define the Matching Logic definedness operator as $\ceil{\phi}_{s_\phi}^s = \exists x @_x^s \phi$. Note that the definedness operator is thus the dual of the universal modality $A$ recalled in the previous section.
\end{remark}

Note that any formula of \ml is a formula of \pl, but the converse does not hold, since a \pl formula might contain nominals or propositional variables.  Let  $Form^0={Form^0_s}_{s\in S}$ be the set of formulas in \pl that does not contain  nominals and propositional variables, i.e. the only variables in these formulas are  state variables. The following remark characterizes  the models of a formula from $Form^0$. 

\begin{remark}
Let $\mathcal{F}=(W, \{R_{\sigma}\}_{\sigma \in \Sigma})$ be an $S$-sorted frame in (hybrid) modal logic and $g: {\rm SVAR}_s \to W$ an assignment function. For any $V_1 \neq V_2$ evaluation functions and any models based on the frame $\mathcal{F}$, $\mathcal{M}_1=(\mathcal{F}, V_1)$ and $\mathcal{M}_2=(\mathcal{F}, V_2)$, we have $\mathcal{M}_1, g, w \mos{s} \varphi$ if and only if  $\mathcal{M}_2, g, w \mos{s} \varphi$ for any $\varphi \in Form^0_s$. In other word, because the evaluation function is defined to evaluate nominals and propositional variables, the satisfiability of formulas which contain only state variables will not be changed in models with the same frame and assignment function, but different evaluation functions.

\end{remark}

For any $s \in S$, $\varphi \in {Form^0_s}$ we define $\f,g,w \mos{s} \varphi$ if and only if $\mathcal{M}, g, w \mos{s} \varphi$ for any $\mathcal{M}$ model based on the frame $\f$ if and only if $\mathcal{M}, g, w \mos{s} \varphi$ for some $\mathcal{M}$ model based on the frame $\f$. Therefore, we use the following notation: $(\f,g) \mos{s} \varphi$ if and only if  $\f,g,w \mos{s} \varphi$  for any $w$ of sort $s$ in any model based on the frame $\f$.

The following definition gives the correspondence between the models of \ml and those of \pl, both logics having the same many-sorted signature.  
%
%
%
%
%
%
%

%
\begin{definition}\label{toml} Let $(S,\Sigma)$ be a many-sorted signature. \\
(1) Let $M$ be a model of \ml and $\rho$ an $M$-valuation. We define the frame $\f _M$ in $\pl$ such that $W_s = M_s$ for any $s \in S$ and $R_{\sigma}w w_1 \ldots w_n$ if and only if $w \in \sigma_M(w_1, \ldots, w_n)$. Moreover, let $ g_s(x)=\rho(x)$ for any $s \in S$ and $x \in {\rm SVAR}_s$. Hence, to any model and valuation $(M, \rho)$ of \ml we associate  a model $(\f, \rho)$ of \pl.

(2) Let  $\ (\f, g)$ be a model of \pl with $\f =(W, \{ R_{\sigma}\}_{\sigma \in \Sigma})$. We define a model in \ml as follows: let $M_s=W_s$ for any $s \in S$, $w \in \sigma_M(w_1, \ldots, w_n)$ if and only if $R_{\sigma}w w_1 \ldots w_n$ and $\rho(x)=g_s(x)$ for any $s \in S$ and $x \in {\rm SVAR}_s$. Hence, to any model $\ (\f, g)\ $ of \pl we can associate a model $(M_{\f}, g)$ of \ml.
\end{definition}
%

In the sequel, we need to speak about satisfiability in \ml  
and satisfiability in \pl. Therefore, to distinguish these two notions, we use $\mos{s}_{\ml}$ when refer to satisfiability in \ml and we use $\mos{s}_{\pl}$ for \pl.

\begin{proposition}\label{prop:equiv}
Let $\varphi \in Form^0_s$. Then
\begin{enumerate}
\item[(1)] $(M, \rho)\mosm{s} \varphi$ if and only if $(\f_M, \rho) \mosp{s} \varphi$
\item[(2)] $(\f, g) \mosp{s} \varphi$ if and only if $(M_{\f}, g)\mosm{s} \varphi$ 
\end{enumerate}

\begin{proof}
We only prove the first item of the proposition by induction over $\varphi$, the other one is similar. 

$\bullet$ $(M, \rho)\mosm{s} x$, where $x \in {\rm SVAR}_s$ iff $\overline{\rho}(x)=M_s=\{w\}$ iff $W_s=M_s$ and $ \rho(x)=w$ iff\\ $\f_M,\rho,w \mosp{s} x$, for any $w \in W_s$ iff $(\f_M,\rho) \mosp{s} x$.

$\bullet$ $(M, \rho)\mosm{s} \neg \varphi$ iff $(M, \rho)\ \not\!\!\mosm{s} \varphi$ iff $(\f_M, \rho)\ \not\!\!\mosp{s} \varphi$ (induction hypothesis) iff \\ $(\f_M, \rho) \mosp{s} \neg \varphi$.

$\bullet$ $(M, \rho)\mosm{s}  \varphi_1 \vee \varphi_2$ iff $(M, \rho)\mosm{s}  \varphi_1$ or $(M, \rho)\mosm{s}   \varphi_2$ iff $(\f_M, \rho) \mosp{s} \varphi_1$ or $(\f_M, \rho) \mosp{s} \varphi_2$ (induction hypothesis) iff $(\f_M, \rho) \mosp{s} \varphi_1 \vee \varphi_2$.

$\bullet$ $(M, \rho)\mosm{s} \sigma(\varphi_1, \ldots, \varphi_n)$ iff $\sigma_M(\overline{\rho}(\varphi_1), \ldots,\overline{\rho}( \varphi_n))=M_s$ iff $\bigcup \{\sigma_M(m_1, \ldots, m_n) \mid m_i \in \overline{\rho}(\varphi_i), $ for any $ i \in [n]  \}=M_s$ iff for any $m \in M_s$ exist $m_1 \in \overline{\rho}(\varphi_1),\ldots ,m_n \in \overline{\rho}(\varphi_n)$ such that $m = \sigma_M(m_1, \ldots,m_n)$ iff  for any $m \in M_s$ exist $m_1 \in M_{s_1},\ldots ,m_n \in M_{s_n}$ such that $R_{\sigma}m m_1 \ldots m_n$ and $\f_M, \rho, m_i  \mosp{s_i} \varphi_i$ for any $i \in [n]$ iff for any $m \in M_s$, $\f_M, \rho,m \mosp{s} \sigma(\varphi_1, \ldots, \varphi_n)$ iff $(\f_M, \rho) \mosp{s} \sigma(\varphi_1, \ldots, \varphi_n)$.

$\bullet$ $(M, \rho)\mosm{s} \forall x \varphi$ iff $\overline{\rho}(\forall x \varphi)=M_s$ iff $\bigcap_{a\in M_{s'}} \{ \overline{\rho'}(\varphi)\mid \mbox{ for all }\rho' \stackrel{x}{\sim} \rho \}=M_s$ iff for all $m \in M_s$ and for all $\rho' \stackrel{x}{\sim} \rho$, $m \in  \overline{\rho'}(\varphi) $ iff for all $\rho' \stackrel{x}{\sim} \rho$, $(M, \rho')\mosm{s} \varphi$ iff for all $\rho' \stackrel{x}{\sim} \rho$, $(\f_M, \rho') \mosp{s} \varphi$ (induction hypothesis) iff $(\f_M, \rho') \mosp{s} \forall x \varphi$.
\end{proof}

\end{proposition}
\begin{theorem}\label{ml1}
For any formula $\varphi$ from \ml, we have $\vds{}_{\ml} \varphi$ iff $\vds{s}_{\pl} \varphi$.
\end{theorem}

\begin{proof}
Let $\varphi$ be a formula of sort $s$ from Matching logic. Then $\varphi \in Form^0_s$. From the Completeness Theorem proved in \cite{rosulics}, $\vds{}_{{}_{\ml}} \varphi$ iff for any model $M$ from \ml, $(M, \rho) \mosm{s}  \varphi$. From the Completeness Theorem proved for \pl, we have $\vds{s}_{{}_{\pl}} \varphi$ iff for any model $\mathcal{M}$ and any $w \in W_s$,  $\mathcal{M}, g, w \mosp{s} \varphi$. Let $M$ be a model from \ml, such that $(M, \rho) \mosm{s}  \varphi$. But $\varphi \in Form^0_s$, so the satisfiability of the formula $\varphi$ is not affected by the evaluation function, but by frame, by assignment function and by world. Therefore, $\vds{s}_{{}_{\pl}} \varphi$ iff for any $w \in W_s$ and any $\f$, $\f, g,w \mosp{s} \varphi$ iff for any $w \in W_s$ and any $\f$, $(\f, g)\mosp{s} \varphi$. But case 1. from Proposition \ref{prop:equiv} tell us that to any model in \ml where  $(M, \rho) \mosm{s}  \varphi$ we can associate a model in \pl such that $(\f_M, \rho) \mosp{s} \varphi$. And case 2 of Proposition \ref{prop:equiv} tell us that for any model in \pl $(\f, \rho) \mosp{s} \varphi$ we have  $(M, \rho) \mosm{s}  \varphi$. Therefore, our proof is completed.
\end{proof}

So far we've remarked that formulas of \ml are particular formulas of \pl and we've analized the satisfaction of such formulas in both logics. In the most general case, a formula from \pl has nominals and propositional variables (that are interpreted as sets that are not necessarily singletons). In the sequel we show how we can represent  any \pl formula in \ml. Following the well-known  theorem of constants, our main steps are the following:
\begin{enumerate}
\item  we  represent the  propositional variables from \pl  as constant operations in \ml; 
\item  we represent the nominals from \pl as constant operations in \ml and, in order to interpret them as singletons, we ask them to satisfy the property of the functional patterns from \ml. 
\end{enumerate} 

We need to recall further definitions from \ml. For each pair of sorts $s_1$ (for the compared patterns) and $s_2$ (for the context in which the equality is used), equality is defined $_{-}$ $ {=^{s_2}_{s_1}}$ $_{-}$ as the following derived construct:
$ \varphi =^{s_2}_{s_1} \varphi'\ \  \equiv \ \  \floor{\varphi \leftrightarrow \varphi'}^{s_2}_{s_1} \ \  \mbox{, where } \varphi, \varphi'\in PATTERN_{s_1}$.

Let $(S, \Sigma)$ be a many sorted signature and assume that 
 $\pl_{(S,\Sigma)}$ \ is the system ${\mathcal H}_{\Sigma}(@_z,\forall)$ as before. We define $\Sigma_{PROP}=\{c_p\mid p\in PROP\}$ and $\Sigma_{NOM}=\{c_i\mid i\in NOM\}$. We set

 $\Sigma' =\Sigma\cup \Sigma_{PROP}\cup \Sigma_{NOM}$ and $\Gamma'=\{\exists x (x=c_i)|i\in NOM\}$. 
 
If $\varphi$ is a formula in $\pl_{(S,\Sigma)}$, let $\vp'$ be the formula obtained by replacing $p$ with $c_p$ for any $p\in PROP$ and $i$ with $c_i$ for any $i\in NOM$. Hence $\vp'$ is a formula in \ml over the signature $(S,\Sigma')$, which will be called $\ml_{(S,\Sigma')}$.

\begin{theorem}\label{ml2}
Let $(S,\Sigma)$ be a many-sorted signature and assume$\vp$ is a formula of sort $s$ in   $\pl_{(S,\Sigma)}$. If $\Sigma'$, $\Gamma'$ and $\vp'$ are defined as above then $\vds{s}_{\pl_{(S,\Sigma)}}\vp$ \ iff\  $\Gamma'\vdash_{\ml_{(S,\Sigma')}}\vp'$.

\end{theorem} 
\begin{proof} Let ${\mathcal M}=(W,\{R_\sigma\}_{\sigma\in \Sigma}, V)$ be a model for $\vp$ in $\pl_{(S,\Sigma)}$. We define 
$M'=(W,\{\sigma_{M'}\}_{\sigma\in\Sigma'})$ such that $\sigma_{M'}$ is defined as in Definition \ref{toml} for $\sigma\in \Sigma$,
 ${c_p}_{M'} = V(p)$ for any $p\in SVAR$ and ${c_i}_{M'}=V(i)$ for any $i\in NOM$. Note that $V(i)$ is a singleton set, so $M'\models_{\ml_{(S,\Sigma')}}\Gamma'$. For any  $g:SVAR\to W$ one can easily prove that 
 ${\mathcal M}, g, w\mos{s}_{\pl_{(S,\Sigma)}} \vp$ for any $w\in W_s$ if and only if $(M,g)\models_{\ml_{(S,\Sigma')}} \vp'$. Conversly, if $(M',\rho)$ is a model for $\vp'$ in 
 $\ml_{(S,\Sigma')}$ such that $M' \models_{\ml_{(S,\Sigma')}} \Gamma'$ then $\overline{\rho}(c_i)$ is a singleton set by \cite[Proposition 5.18]{rosu}, so we can safely define $V(i)=\overline{\rho}(c_i)$ for any $i\in NOM$ and $V(p)=\overline{\rho}(c_p)$ for any $p\in PROP$. If ${\mathcal M}=(M', \{R_\sigma\}_{\sigma\in \Sigma}, V)$ where $R_\sigma$ is defined as in Definition \ref{toml}. One can easily see that 
 $M'\models_{\ml_{(S,\Sigma')}}\vp'$ if and only if ${\mathcal M}, \rho, w\mos{s}_{\pl_{(S,\Sigma)}}\vp$ for any $w\in M_s$. The intended syntactical connection follows using the completness theorems for \ml and \pl. 
\end{proof}

\section{Conclusions}

The results proved in Section \ref{sec3} allow the transfer of results between  many-sorted hybrid modal logic and Matching logic. Note that, both in this paper, as well as in \cite{rosulics,rosu}, there are two pairs of systems we can consider, the connection being stated by Theorem \ref{ml1} and Theorem \ref{ml2}: ${\mathcal H}_{\Sigma}(@_z,\forall)$  is related to Matching logic with Definedness \cite{rosu}, while 
${\mathcal H}_{\Sigma}(\forall)$ is related to Matching logic without Definedness \cite{rosulics}.   

While Matching logic is a young logic for program verification, the hybrid modal logic is quite established, with roots go back to the work of Prior in the 50's \cite{prior}. As we proved in this paper, they are strongly connected and the connection goes both ways.  At the same time, each system has its peculiarities, an important distinction being the local (in the modal case) versus global (in the case of Matching logic) approach to deduction.    Modal logic in general and hybrid logic in particular has a plethora of applications, both theoretical and practical. Matching logic supports the development of the ${\mathbb K}$ framework, leading not only to formal specification, but also to concrete implementations. We hope that the interaction between this two approaches  will be of further interest for both systems and we plan to further investigate it in the future.

%

\nocite{*}
\bibliographystyle{eptcs}

\end{document}